\documentclass{llncs}
\usepackage{graphicx}
\usepackage[utf8]{inputenc}
\usepackage[T1]{fontenc}
\usepackage[english]{babel}
\usepackage{amsmath,amssymb}
\usepackage{graphicx}
\usepackage{color}
\usepackage{url,ifthen}
\usepackage{fancyhdr}
\usepackage{mathtools}
\usepackage{algorithm}
\usepackage{algpseudocode}
\usepackage[all,dvips]{xy}
\usepackage{float}
\usepackage[normalem]{ulem}

\newcommand{\C}{\mathcal{C}}
\newcommand{\D}{\mathcal{D}}
\newcommand{\I}{\mathcal{I}}
\newcommand{\F}{\mathbb{F}}

\newcommand{\word}[1]{\mathbf{#1}}

\newcommand{\bv}{\word{b}}
\newcommand{\cv}{\word{c}}

\newcommand{\gv}{\word{g}}
\newcommand{\hv}{\word{h}}
\newcommand{\mv}{\word{m}}

\newcommand{\sv}{\word{s}}

\newcommand{\xv}{\word{x}}
\newcommand{\yv}{\word{y}}

\newcommand{\zz}{\word{0}}

\newcommand{\mat}[1]{\boldsymbol{#1}}

\newcommand{\Gm}{\mat{G}}
\newcommand{\Hm}{\mat{H}}

\newcommand{\Lm}{\mat{L}}
\newcommand{\Mm}{\mat{M}}

\newcommand{\Qm}{\mat{Q}}

\newcommand{\Tm}{\mat{T}}
\newcommand{\Um}{\mat{U}}

\newboolean{anonymous}
\setboolean{anonymous}{false}
\begin{document}
		\title{Generalized Subspace Subcodes in the Rank Metric}
	
	\ifthenelse{\boolean{anonymous}}
	{ 
		\author{
			Anonymized for submission
		}
		\institute{
		}
	} 
	{ 
		\author{ 
                Ousmane Ndiaye\inst{1}
                \and
                Peter Arnaud Kidoudou\inst{2}
                \and
                Herv\'e Tale Kalachi\inst{3}
		}
		\institute{
            Université Cheikh Anta Diop de Dakar, FST, DMI, LACGAA,\\ Senegal,\\\email{ousmane3.ndiaye@ucad.edu.sn}
            \and
              université Marien Ngouabi (UMNG), Faculty of Science and  Technology\\
            \email{peter.kidoudou@umng.cg}
                \and
			Computer Engineering Department, National Advanced School of Engineering of Yaound\'e, University of Yaoundé 1, Yaoundé,
Cameroon,\\
			\email{hervekalachi@gmail.com }
		}
		
	} 
	
	\date{\today}
	\maketitle
	
	\begin{abstract}
Rank-metric codes were studied by E. Gabidulin in 1985 after a brief introduction by Delsarte in 1978 as an equivalent of Reed-Solomon codes, but based on linearized polynomials. They have found applications in many areas, including linear network coding and space-time coding.
 They are also used in cryptography to reduce the size of the keys compared to Hamming metric codes at the same level of security. However, some families of rank-metric codes suffer from structural attacks due to the strong algebraic structure from which they are defined. 
 It therefore becomes interesting to find new code families in order to address these questions in the landscape of rank-metric codes.
 \par In this paper, we provide a generalization of Subspace Subcodes in Rank metric introduced by Gabidulin and Loidreau. We also  characterize this family by giving an algorithm which allows to have its generator and parity-check matrices based on the associated extended codes. We have also studied the specific case of Gabidulin codes whose underlying decoding algorithms are known. Bounds for the cardinalities of these codes, both in the general case and in the case of Gabidulin codes, are also provided. 
\end{abstract}

\section{Introduction}
Rank-metric codes were studied by E. Gabidulin in 1985 \cite{gabidulin1985theory} after a brief introduction by Delsarte in 1978 \cite{DELSARTE1978226} who gave a description based on finite fields and its properties. By construction, Gabidulin's codes are very close to Reed-Solomon \cite{Reed1960PolynomialCO} codes, since codewords are obtained by evaluation of $q-$polynomials on a support included in an extension of degree $m$ of $\F_q$.
Rank-metric codes have found applications in network coding \cite{El2018efficient}, for example, when the transmitter and receiver are oblivious to the inner workings and topological network. Rank-metric codes has been also used in the theory of space-time codes, introduced by Lu and Kumar from Gabidulin codes \cite{1424310}. 

Furthermore, today's cryptography uses number theory problems (factorization, discrete log, ...) to secure our information. However, it is well known from \cite{S94a} that all these problems can be efficiently solved using quantum computers. 
The theory of error correcting codes is a serious candidate that offers perspectives to face this machine via the Syndrome Decoding problem in the Hamming metric or in the Rank metric. 
They have been very used in cryptography these last decades to provide cryptographic primitives for encryption, signature, hashing or pseudo-random number generation. The main reason for using rank-metric codes in cryptography is the possibility of reducing the size of the keys compared to Hamming metric codes at the same level of security. 
Rank metric based cryptography began with the GPT cryptosystem and its variants \cite{gabidulin1991ideals,Berger2005HowTM,Gabidulin2008attacks,Loidreau2010designing,Rashwan2010smart,Gabidulin2008attacks,Rashwan2011security} based on Gabidulin codes \cite{gabidulin1985theory}, which are rank-metric analogues of Reed-Solomon codes.
Unfortunately, as in the case of Reed-Solomon codes, the strong algebraic structure of these codes has been successfully exploited to attack the original GPT cryptosystem and its variants in a series of works initiated by Gibson and Overbeck   \cite{G95,G96,overbeck2005new,Otmani2018improved,Horlemann2018extension,Kalachi2022failure}.
Recently, some public key cryptosystems based on rank-metric codes \cite{is2019rollo,melchor2017rank} were candidates at the NIST competition for Post-Quantum Cryptography and up to the second round. But again, they were all eliminated after the second round of the competition due to security defects. Considering the fact that the very old McEliece cryptosystem that uses Subspace Subcodes of Generalized Reed-Solomon codes is still secure and currently at the fourth found of the NIST competition, it becomes  theoretically interesting to study Subspace Subcodes in the rank-metric.

The notion of Subspace Subcode is used to denote a Subcode whose components of each codeword belong to the same vector subspace. This notion was first introduced in the Hamming metric for Reed-Solomon codes by Hattori, McEliece, and Solomon, G. \cite{705564}. A few years later, the same notion was introduced for rank-metric codes by Gabidulin and Loidreau with applications to cryptography \cite{GabLoidreau2008,Berger2005HowTM}. Recently, Berger, Gueye and Klamti \cite{BergGSS2019} proposed a generalization of these Subcodes in the Hamming metric by allowing the components of each codeword of the Subcode to be in different subspaces. This previous work was followed by an article of Berger, Gueye, Klamti and Ruatta \cite{BGKR19}, proposing a cryptosystem based on quasi–cyclic Subcodes of Subspace Subcode of Reed-Solomon codes.

This paper introduces and characterizes the family of Generalized Subspace Subcodes of a rank-metric code by giving an algorithm which allows to construct their generator and parity-check matrices based on the associated extended codes. Bounds of the cardinalities of these codes both in the general case and in the case of Gabidulin codes are also provided as a generalization of the results obtained in \cite{GabLoidreau2008}. The work ends by a focus on the specific case of Gabidulin codes whose decoding algorithms are known.

The rest of the document is organized as follows, we start in section \ref{sec:prel} with some preliminaries to identify the generalities related to rank-metric codes. Section \ref{sec:gssrm} presents some results on Generalized Subspace Subcodes in the rank-metric and a new result on the lower bound of their size. We give our main results on Generalized Subspace Subcodes of Gabidulin codes in section \ref{sec:gab} by giving bounds of their size and an algorithm to construct directly a generator matrix. Section \ref{sec:parent} deals with the parent Codes of Gabidulin RGSS codes for decoding subcodes as alternative to using the decoding algorithm of the supercode. In section \ref{sec:application} we give some directions for their potential applications in cryptology and how it can reduce the key sizes. Finally, section \ref{sec:conclusion} concludes the paper and gives some perspectives.

\section{Preliminaries} \label{sec:prel}
\subsection{Rank-Metric Codes}
The \emph{rank weight} $w_{R}(\xv)$ of a word $\xv=(x_1, \dots, x_n) \in \mathbb{F}_{q^m}^n$ in an extension field $\mathbb{F}_{q^m}/\mathbb{F}_{q}$ is defined by the maximal number of its elements that are linearly independent over the base field $\mathbb{F}_q$, where $m$ and $n$ are positive integers and $q$ a power of a prime number. The \emph{rank distance} $d_{R}$ between two words is defined by the rank weight of their difference, i.e. $d_{R}(\xv,\yv)=w_{R}(\xv-\yv)$. It is well known that $d_{R}$ has the properties of a metric on $\mathbb{F}_{q^m}^n$. \\
\begin{definition} Let $\xv=(x_1, \dots, x_n) \in \mathbb{F}_{q^m}^n$ and $\{b_1, b_2, ..., b_m\}$ a basis of $\F{q^m}$ over $\F_q$. We can write $x_j = \sum\limits_{i=1}^m x_{ij} b_i \in \F_{q^m}$ for each $j = 1, \cdots,n$ with $x_{ij} \in \F_q$. The rank weight $w_{R}(\xv)$ of $\xv$ is defined as the rank of the matrix $M_{\xv}=(x_{ij}) \in \F_q^{m\times n}$.
\end{definition}
In the sequel, a linear code of length $n$ and dimension $k$ (also called $[n,k]-$linear code or simply $[n,k]-$code) over $\mathbb{F}_{q^m}$ will denote a $k-$dimensional subspace of the $n-$dimensional vector space $\mathbb{F}_{q^m}^n$. A rank-metric code is then a linear code endowed with the above metric, called rank-metric. Given a rank-metric code $\C \subset \F_{q^m}^n$, its minimum rank distance is
$$d_R(\C)= \min\limits_{\cv_1\neq \cv_ 2 \in \C }w_{R}(\cv_1-\cv_2)=\min\limits_{\cv \in \C , \cv \neq \zz }w_{R}(\cv).$$




If $\C \subset \mathbb{F}_{q^m}^n$ is a $k-$dimensional rank-metric code with minimum rank distance $d_R(\C)$, it is said to be a $[n, k, d=d_R(\C)]_{q^m}-$code. The
parameters of a $[n, k, d=d_R(\C)]_{q^m}-$code are related by an equivalent of the Singleton bound for the rank distance, see \cite{gabidulin1985theory}:
$$Card(\C)\leq q^{\min(m(n-d+1),n(m-d+1))}.$$ 
Where $Card$ is the cardinality of $\C$.
Furthermore, a code satisfying the equality $Card(\C) = q^{\min(m(n-d+1),n(m-d+1))}$ is called a Maximum Rank Distance (MRD) code.

\subsection{Punctured and Shortened Codes}
Punctured and shortened codes of a given code are very important derivative codes to characterize certain subcodes, but also to mount structural attacks against code-based cryptosystems \cite{AttackwildMcEliece2017}. We are going to use them on extended codes to characterize Subspace Subcodes in the rank-metric. We recall here their definitions and some needed properties.

\begin{definition}(Punctured code)\\
	Let $\C \subset \F_q^n$ be a linear code, $\I \subseteq \{1, .., n\}$ a set of coordinates.\\
The punctured code of $\C$ on $\I$ denoted by $Punct_{\I}(\C)$ is defined by :
$$Punct_{\I}(\C)=\{ (c_i)_{i \in \{1, .., n\}\setminus \I} ~|~ c \in \C \}.$$
This code is of length $n-|\I|$.
\end{definition}
\begin{proposition}
  Let $\C \subset \F_q^n$ be a $[n,k,d]-$code, $\I \subseteq \{1, .., n\}$. Then $Punct_{\I}(\C)$  is an  \\
  $[n-|\I|, k', d']-$code such that:
  $$ k-|I| \leq k' \leq k ~and~ d-|I| \leq d' \leq d.$$
\end{proposition}

\begin{definition}(Shortened Code)\\
	Let $\C \subset \F_q^n$ be a linear code, and $\I \subseteq \{1, .., n\}$ a set of coordinates.\\
The  Shortened code of $\C$ on $\I$ denoted by $Short_{\I}(\C)$ is defined by :
$$Short_{\I}(\C)=Punct_{\I}(\{ c \in \C ~|~ c_i=0, ~\forall ~i \in \I \}).$$
This code is of length $n-|\I|$.
\end{definition}
\begin{proposition}
  Let $\C \subset \F_q^n$ be a $[n,k,d]-$code and $\I \subseteq \{1, .., n\}$. Then $Short_{\I}(\C)$ is a  \\
  $[n-|\I|, k', d']-$code satisfying:
  $$ k-|I| \leq k' \leq k ~and~  d \leq d'.$$
\end{proposition}
\begin{theorem}(Link between Shortened and Punctured Codes [\cite{huffman_pless_2003}, Theorem 1.5.7])\label{shorpunct}\\
  Let $\C \subset \F_q^n$ be a $[n,k,d]-$code and $\I \subseteq \{1, .., n\}$. We have
   \begin{enumerate}
     \item $Short_{\I}(\C^{\bot})=Punct_{\I}(\C)^{\bot}$
     \item $Short_{\I}(\C)^{\bot}=Punct_{\I}(\C^{\bot})$
   \end{enumerate}
\end{theorem}
This link will allow us to have a generator matrix of a shortened code (subcode) of a code from the punctured code of its dual code.
\section{Generalized Subspace Subcodes in the Rank Metric}\label{sec:gssrm}
Subspace Subcodes in the rank-metric were introduced by Gabidulin and Loidreau in \cite{GabLoidreau2008}, where decoding methods were also proposed. In this section, we give a generalization of Subspace Subcodes in the rank-metric.
\begin{definition}[\cite{GabLoidreau2008}]\label{defss}
Let $\mathcal{C}$ be a $[n, k, d=d_R(\C)]_{q^m}-$code, and $V$ a $\F_q-$vector subspace of $\F_{q^m}$ with dimension $s \leq m$. The Subspace Subcode of $\C$, with respect to $V$, is the $\F_q-$vector space $(\C_{|V}) = \C \cap V^n.$
\end{definition}
The elements of $(\C_{|V})$ are codewords whose components lie in the alphabet formed by the subspace $V$. In general, $(\C_{|V})$ is not $\F_{q^m}-$linear, but it is $\F_q-$linear and also linear over some intermediate extension depending on $V$. This code also corresponds by projection to the Subgroup Subcode \cite{382025} on the alphabet $\F_q^s$.\\
In the following, we introduce a generalization of definition \ref{defss} by allowing the choice of different vector subspaces for each coordinate.
\begin{definition}\label{defgss}
Let $\C$ be a $[n, k, d=d_R(\C)]_{q^m}-$code, and  $V_1,...,V_n$ a series of $n$ $\F_q-$vector subspaces of $\F_{q^m}$ with dimensions respectively $s_1$, $s_2$, ..., $s_n$. Set  
$W={\displaystyle \prod_{i=1}^{n} V_i}$, the Cartesian product of $V_1,...,V_n$. The Rank Generalized Subspace Subcode of $\C$ with respect to $W$ is the $\F_q-$vector space $RGSS_W(\C) = \C \cap W.$
\end{definition}

Remark that all the $V_i$'s are equal in definition \ref{defss}, contrary to what we have in the above definition \ref{defgss}. \\
Let $B = \{b_1, b_2, ..., b_m\}$ be a basis of $\F_{q^m}$ as a $\F_q-$vector space. We consider the map $\phi_B : \F_{q^m} \longrightarrow \F^m_q$, mapping each element of $\F_{q^m}$ to its $\F_q-$coordinates in the basis $B$. That is to say, for any $x = \sum\limits_{i=1}^m x_i b_i \in \F_{q^m}$ (with $x_i \in \F_q$),  $\phi_B(x) = (x_1, x_2,...,x_m)$. Given $n$ $\F_q-$bases $B_1,..., B_n$ of $\F_{q^m}$, this map can be extended to $\F_{q^m}^n$ by the isomorphism:\\
\begin{align*}
Exp_{(B_i)_i}\colon \F_{q^m}^n & \longrightarrow\F^{mn}_q\\
(c_1, c_2,...,c_n)& \longmapsto (\phi_{B_1}(c_1), \phi_{B_2}(c_2),...,\phi_{B_n}(c_n)),
\end{align*}



$Exp_{(B_i)_i}$ is called the expansion function and, applying it to all codewords of a code $\C$, this gives a new linear code of length $nm$ called Expanded Codes.  
\begin{definition}
Let $\C$ be a linear code of length $n$ and dimension $k$ over $\F_{q^m}$. The
expanded code of $\C$ also called ($q-$ary image of $\C$) with respect to the $n$ bases $B_1,..., B_n$ of $\F_{q^m}$ is the $\F_{q}-$linear code defined by:
    $$Im_{q,(B_i)_i}(\C)=Exp_{(B_i)_i}(\C)$$
\end{definition}
Given a $[n,k]-$code $\C$ over $\F_{q^m}$, the following proposition shows how to compute a generator matrix or a parity-check matrix of the $q-$ary image of $\C$. Before stating this proposition, it is important to note that for the extended code, while the components of the codewords are expressed in $n$ different bases, it is assumed for measures of simplicity and efficiency of encoding that all the components (coordinates) of the message to be encoded are expressed in the same base $B$. Consequently, each component $i$ is obtained by a $\F_q-$linear application of $\F_{q^m}$ to $\F_{q}^m$ in the bases $B$ and $B_i$.
\begin{proposition}\label{ExtGenMat}
Let $\C$ be a $[n,k]-$code over $\F_{q^m}$ 
    \begin{enumerate}
        \item If $\Gm=\begin{pmatrix} \gv_{1}\\ 
        \gv_{2}\\ 
		 \vdots \\ 
		\gv_{k}\\ 
	\end{pmatrix} \in \F_{q^m}^{k\times n}$ is a generator matrix of $\C$ then, the matrix $\hat{\Gm}_{(B_j)_j}^B$ defined by
 \begin{equation}
     \hat{\Gm}_{(B_j)_j}^B=ExpMat_{(B_j)_j,n}^B(\Gm)=\begin{pmatrix} 
        Exp_{(B_j)_j}(b_1\gv_i)\\ 
        Exp_{(B_j)_j}(b_2\gv_i)\\ 
		 \vdots \\ 
	Exp_{(B_j)_j}(b_m\gv_i)\\ 
	\end{pmatrix}_{1\leq i\leq k} \in \F_{q}^{mk\times mn}
 \end{equation}
is a generator matrix of the expanded code of $\C$ over $\F_q$ from the basis $B = \{b_1, b_2, ..., b_m\}$ with respect to the $n$ bases $B_1,..., B_n$ of $\F_{q^m}.$

    \item If 
        $\Hm=\begin{pmatrix} \hv_{1}\\ 
        \hv_{2}\\ 
		 \vdots \\ 
		\hv_{n}\\ 
	\end{pmatrix}^T \in \F_{q^m}^{(n-k)\times n}$ is a parity-check matrix of $\C$, then the matrix $\hat{\Hm}^B_{(B_j)_j}$ defined by
       \begin{equation}
           \hat{\Hm}^B_{(B_j)_j}=ExpMat_{(B_j)_j,n-k}^B(H^T)^T=\begin{pmatrix} 
        Exp_{(B_j)_j}(b_1 \hv_i)\\ 
        Exp_{(B_j)_j}(b_2 \hv_i)\\ 
		 \vdots \\ 
	Exp_{(B_j)_j}(b_m \hv_i)\\ 
	\end{pmatrix}_{1\leq i\leq n}^T \in \F_{q}^{m(n-k)\times mn}
       \end{equation}
 
is a parity-check matrix of the expanded code of $\C$ over $\F_q$ from the basis $B = \{b_1, b_2, ..., b_m\}$ with respect to the $n-k$ bases $B_1,..., B_{n-k}$ of $\F_{q^m}$.
    \end{enumerate}
\end{proposition}

\begin{proof}
For the first point, let us consider for any line $i \in \{1, \cdots k\}$ of $\Gm$, the map 
\begin{align*}
f_{\gv_i} \colon \F_{q^m} & \longrightarrow\F_{q^m}^{n}\\
x &\longmapsto f_{\gv_i}(x) = x \gv_i,
\end{align*}
Since $Exp_B:=\phi_B$ is an $\F_q-$isomorphism from $\F_{q^m}$ to $\F_q^m$, we consider the $\F_q-$linear map $\theta_{\gv_i} = Exp_{(B_j)_j} \circ f_{\gv_i} \circ Exp_{B}^{-1}$ from $\F_q^m$ to $\F_q^{mn}$. Let $B^* = (b_1^*,...,b_m^*)$ be a basis of $\F_q^m$ such that for any $i \in \{1, \cdots, m\}$,  $b^*_i = Exp_{B}(b_i)$.\\ As the encoding is considered as a vector-matrix product, the matrix of $\theta_{\gv_i}$ in the basis $B^*$ is given by a matrix $\Mm_{\gv_i} \in \F_q^{m\times mn}$ such that:
\[
\Mm_{\gv_i} =    
        \begin{pmatrix} 
        \theta_{\gv_i}(b^*_1)\\ 
        \theta_{\gv_i}(b^*_2)\\ 
		 \vdots \\ 
	\theta_{\gv_i}(b^*_m) 
	\end{pmatrix} 
         = 
        \begin{pmatrix} 
        Exp_{(B_j)_j} \circ f_{\gv_i} \circ Exp_{B}^{-1}(b^*_1)\\ 
        Exp_{(B_j)_j} \circ f_{\gv_i} \circ Exp_{B}^{-1}(b^*_2)\\ 
		 \vdots \\ 
	Exp_{(B_j)_j} \circ f_{\gv_i} \circ Exp_{B}^{-1}(b^*_m)
	\end{pmatrix}
        = 
        \begin{pmatrix} 
        Exp_{(B_j)_j}(b_1 \gv_i)\\ 
        Exp_{(B_j)_j}(b_2 \gv_i)\\ 
		 \vdots \\ 
	Exp_{(B_j)_j}(b_m \gv_i) 
	\end{pmatrix} 
\]
To finish, we have by definition
\[
\begin{tabular}{ccl}
    $Exp_{(B_j)_j} (\C)$ & $=$ & $\{Exp_{(B_j)_j}(\mv \Gm), m \in \F_{q^m}^k \}$ \\
                      & $=$ & $\{Exp_{(B_j)_j}(\sum_{i=1}^{k} m_i \gv_i), \mv = (m_1, \cdots, m_k) \in \F_{q^m}^k \}$  \\
                      &  $=$  & $\{\sum_{i=1}^{k} Exp_{(B_j)_j} \circ f_{\gv_i} ( m_i), \mv = (m_1, \cdots, m_k) \in \F_{q^m}^k \}$ \\
                      &  $=$  & $\{\sum_{i=1}^{k} \theta_{\gv_i} \circ Exp_{B} ( m_i), \mv = (m_1, \cdots, m_k) \in \F_{q^m}^k \}$ \\
                      &  $=$  & $\{ \sum_{i=1}^{k} \theta_{\gv_i} ( \mv_i^ \prime), \mv^\prime = (\mv_1^\prime, \cdots, \mv_k^\prime) \in (\F_{q}^{m})^{k} \}$\\
                      &  $=$  & $\{ \sum_{i=1}^{k}  \mv_i^ \prime \Mm_{\gv_i}, \mv^\prime = (\mv_1^\prime, \cdots, \mv_k^\prime) \in (\F_{q}^{m})^{k} \}$\\
                      &  $=$  & $\{  \mv^\prime\cdot \begin{pmatrix} 
                                \Mm_{\gv_1} \\ 
                                \Mm_{\gv_2}\\ 
                                 \vdots \\ 
                                \Mm_{\gv_k} 
                                \end{pmatrix}, ~~ \mv^\prime \in \F_{q}^{km} \}$\\
                      &  $=$  & $<\begin{pmatrix} 
                                \Mm_{\gv_1} \\ 
                                \Mm_{\gv_2}\\ 
                                 \vdots \\ 
                                \Mm_{\gv_k} 
                                \end{pmatrix}>_{\F_q}$
\end{tabular}
\]
The proof of the second point  of the proposition is similar. Indeed, by considering only the $n-k$ bases $B_1, ..., B_{n-k}$ of $\F_{q^m}$, let us apply $Exp_{(B_j)_j}$ on a syndrome $\sv$ given by $\sv^T= \yv \cdot H^T$. We have\\

\begin{tabular}{lll}
   $Exp_{(B_j)_j} (\sv^T)$  & $=$  & $Exp_{(B_j)_j} (\yv \cdot H^T)$ \\
      & $=$ & $Exp_{B} (\yv)\cdot ExpMat_{(B_j)_j,n-k}^B(H^T)$ \\
      & $=$ & $ExpMat_{(B_j)_j,n-k}^B(H^T)^T \cdot Exp_{B} (\yv)^T $\\
      & $=$ &  $\hat{\Hm}^B_{(B_j)_j} \cdot Exp_{B} (\yv)^T$
\end{tabular}

\qed
\end{proof}
One can remark that this proposition is a general case of \cite[Theorem 1]{5695134}, that can be obtained from the above proposition by taking $B_1 = B_2= \cdots = B_n$. Furthermore, using the notations in the above proof, one can remark that for any $i \in \{1, \cdots, k\}$, $f_{\gv_i} = (f_{g_{i1}}, f_{g_{i2}}, \cdots, f_{g_{in}} )$. So, the matrix  $\hat{\Gm}_{(B_j)_j}^B$ can be written as follows : 

\[
    \hat{\Gm}_{(B_j)_j}^B = (\Mm_{g_{ij}})_{1 \leq i\leq k, 1 \leq j\leq n}  = \begin{pmatrix} 
        \Mm_{g_{11}} & \cdots & \Mm_{g_{1n}}  \\ 
        \Mm_{g_{21}} & \cdots & \Mm_{g_{2n}}  \\ 
		 \vdots   & \ddots & \vdots    \\ 
        \Mm_{g_{k1}} & \cdots & \Mm_{g_{kn}}   
	\end{pmatrix} \in (\F_{q}^{m\times m})^{k\times n}
\]
where the matrix $\Mm_{g_{ij}}^T =\mathcal{M}_{B,B_j}(\theta_{g_{ij}})$ is the matrix of the linear map $\theta_{g_{ij}}$ from the basis $B$ to the basis $B_j$, for $i \in \{1, \cdots k\}$ and $j \in \{1, \cdots n \}$. If $B_1 = B_2= \cdots = B_n = B$, we define $\hat{\Gm}_B := \hat{\Gm}_{(B_j)_j}^B$ and also $Exp_{B,n} := Exp_{(B_j)_j}$.

To build a subcode on the subspaces, the basic idea is to extend the code on a primitive field $\F_q$, then carry out changes of bases on each window of length $m$ of the dual and finally move on to puncturing to drop unnecessary columns to obtain the dual of the subcode sought.
Before giving a very useful proposition for the sequel, let us start by the following definition and notations.
\begin{definition}
Let $U \subseteq \{1,2,...,nm\}$.
We denote by $S_U$ (respectively $P_U$) the operation of shortening (respectively puncturing) the $q-$ary image of $\C$ on positions $I_U= \{1,2,...,nm\} \setminus U$, i.e. 
$S_{U}(\C) := Short_{I_U}(Im_{q,(B_j)_j}(\C))$ and $P_U(\C) := Punct_{I_U}(Im_{q,(B_j)_j}(\C))$. 
\end{definition}

We then have the following proposition.  
\begin{proposition} \label{short}
Let $\C$ be a linear code of length $n$ over $\F_{q^m}$, $B = \{b_1, b_2, ..., b_m\}$ be a $\F_q-$basis of $\F_{q^m}$, $B_s = \{ b_1,..., b_s \}$ and $V=<b_1,..., b_s>$ an $s-$dimensional $\F_q-$subspaces of $\F_{q^m}$. Consider the set $U=u_{1}\cup ... \cup u_{n} \subseteq \{1,2,...,nm\} $ such that for $1 \leq j\leq n$, $u_j=\{(j-1)m+1,..., (j-1)m+s\}$.
The code $S_U(\C)$ is the $q-$ary image of the Subspace Subcode of $\C$ over $V$ with respect to the basis $B_s$.
\end{proposition}

\begin{proof}
Let $\yv=(y_{1,1},..., y_{s,1}, ..., y_{1,j},..., y_{s,j}, ..., y_{1,n},..., y_{s,n}) \in S_U(\C)$ then there exists  $c$ in $Im_q(\C) = Exp_{B,n}(\C)$ such that $\yv=Punct_{I_U}(c)$ and satisfying $c_i=0$ for any $i \notin U$. 
As $V=\phi_{B}^{-1}(\{(x_1,..., x_{s}, 0,...,0) |  x_1,..., x_{s} \in \F_q \})$,
 $Exp_{B,n}^{-1}(c) \in \C\cap V^n$ and so, $c \in Exp_{B,n}(\C \cap V^n)$. Applying $Punct_{I_U},$ we have 
 
 $$\yv \in Punct_{I_U}(Exp_{B,n}(\C \cap V^n)) = Exp_{B_s,n}(\C \cap V^n))$$ 

 Conversely, it is obvious that $Exp_{B_s,n}(\C \cap V^n)) \subset S_U(\C).$
 That means
$$S_U(\C)=Exp_{B_s,n}(\C \cap V^n))=Im_{q,B_s}(\C \cap V^n) $$ \qed 
\end{proof} 

The above proposition shows that shortening the $q-$ary image of a code $\C$ gives the $q-$ary image of a given Subspace Subcode of $\C$. 
\begin{theorem}\cite[lemma 34]{couvreur:hal-02938812}  \label{theo34}
Let $\C$ be a linear code of length $n$ over $\F_{q^m}$, $B = \{b_1, b_2, ..., b_m\}$ be a basis of $\F_{q^m}$. Let $(\Qm_j)_j \in (GL_m(\F_q))^n$. The following equalities hold.
$$Im_{q,(\Qm_j^{-1}B)_j}(\C)=Exp_{(\Qm_j^{-1}B)_j}(\C)=Exp_{B,n}(\C)\cdot\begin{pmatrix} 
        \Qm_1&&\\ 
		 &\ddots& \\ 
	&& \Qm_n\\ 
	\end{pmatrix},$$
 $$\hat{\Gm}_{(B_j)_j}^B=\hat{\Gm}_B\cdot\begin{pmatrix} 
        \Qm_1&&\\ 
		 &\ddots& \\ 
	&& \Qm_n\\ 
	\end{pmatrix}, and$$

 $$\hat{\Hm}^B_{(B_j)_j}=\hat{\Hm}_B\cdot\begin{pmatrix} 
        (\Qm_1^{-1})^T&&\\ 
		 &\ddots& \\ 
	&& (\Qm_n^{-1})^T\\ 
	\end{pmatrix}$$
 
\end{theorem}

\begin{corollary}
Let $\C$ be a linear code of length $n$ over $\F_{q^m}$, $B = \{b_1, b_2, ..., b_m\}$ be a basis of $\F_{q^m}$ and $V=<d_1,..., d_s>_{\F_q} $ a $s-$dimensional subspaces of $\F_{q^m}$. The $q-$ary image of the Subspace Subcode $(\C_{|V}) = \C \cap V^n$ with respect to $B$, can be expressed as $S_U(\C)$ on any completed basis $D = \{d_1, d_2, ...,d_s, ...,  d_m\}$ of $\F_{q^m}$
\end{corollary} 

\begin{proof}
   According to Proposition \ref{short} and letting $D_s=\{d_1,d_2,...,d_s\}$, we have $S_U(\C) = Exp_{D_s,n}(\C \cap V^n) = Punct_{I_U}(Exp_{D,n}(\C \cap V^n))$. As $QD=B$ with $Q$ be the change-of-basis matrix from the basis $B$ to the basis $D$, then
$$ {\tiny S_U(\C)=Punct_{I_U} \left(Exp_{Q^{-1}B,n}(\C \cap V^n)\right) = Punct_{I_U} \left(Exp_{B,n}(\C \cap V^n)\cdot\begin{pmatrix} 
        Q&&\\ 
		 &\ddots& \\ 
	&& Q\\ 
	\end{pmatrix} \right) .}$$
Let $\Qm_{u} \in \F_{q}^{m\times s}$ be the sub-matrix of $\Qm$, obtained by keeping all the columns of $\Qm$ with indexes in $u=\{1,...,s\}$. Then,
$$S_U(\C) = Exp_{B,n}(\C \cap V^n)\cdot\begin{pmatrix} 
        \Qm_{u} &&\\ 
		 &\ddots& \\ 
	&& \Qm_{u} \\ 
	\end{pmatrix} .$$ \qed 
\end{proof}
More generally, we have the following corollary.
\begin{corollary}
    Let $\C$ be a linear code of length $n$ over $\F_{q^m}$ and $B$ be a $\F_q-$basis of $\F_{q^m}$. For any $i \in \{1,\cdots m \}$, let  $V_i=<d_{1,i},..., d_{s_i,i}>$ be a $s_i-$dimensional $\F_q-$subspaces of $\F_{q^m}$ and $W={\displaystyle \prod_{i=1}^{n} V_i}$. Then, the $q-$ary image of $RGSS_W(\C) = \C \cap W$ can be expressed as $S_U(\C)$ from any basis $B = \{b_1, b_2, ..., b_m\}$ with respect to the $n$ completed bases $B_1,..., B_n$ from those of $V_i$ respectively.
\end{corollary} 

By using the relation between shortened codes and punctured codes, one can build the generator matrix of a Subcode on the subspace. A generator matrix of $Exp_{(\Qm_j^{-1}B)_j}(\C \cap V^n)$ can be computed from its parity check matrix defined by $Punct_{I_U}(\hat{\Hm}_{(\Qm_j^{-1}B)_j})$ according to the theorem \ref{shorpunct}.\\

In what follows we give the lower bound of the cardinality of a generalized subcode on the subspaces in metric rank according to the dimensions of the subspaces.
\begin{proposition}\label{prop1}
Let $\mathcal{C}$ be a $[n, k, d=d_R(\C)]_{q^m}-$code and $V_1,...,V_n$ be a series of $n$ $\F_q-$subspaces of $\F_{q^m}$ with dimensions respectively $s_1, \cdots, s_n \leq m$. Set  
$W={\displaystyle \prod_{i=1}^{n} V_i}$, constituted of $n-$tuples with the $i-$th coordinate in $V_i$.
If $\sum\limits_{i=1}^{n}s_i - m(n-k) > 0$ then:
$$q^{\sum\limits_{i=1}^{n}s_i - m(n-k)} \leq |RGSS_W(\C)|$$
\end{proposition}
\begin{proof}
    Without altering the generality, assume that $s_1\leq s_2 \leq ... \leq s_n$. So, up to isomorphisms, we have $V_1 \subset ... \subset V_n$, even if it means permuting the columns and obtaining isometrically equivalent $\F_q-$linear codes. So there exists a series of bases $b_1, ..., b_n$ for $V_1,...,V_n$ respectively, satisfying $b_1 \subset ... \subset b_n$ and $b_i=\{\beta_1, \beta_2,..., \beta_{s_i}\}$.
    Let   $c=(c_1,...,c_n) \in W$. Then, $$c_i=u_{1i}\beta_1+...+u_{s_ii}\beta_{s_i}+0\beta_{s_i+1}+...+0\beta_{s_n}$$
    which translates to
    \[
    c=b_n \Um =(\beta_1, \beta_2,..., \beta_{s_n})\Um \text{
    with } U=(u_{i,t})_{i=1,t=1}^{s_n,n} \in \F_{q}^{s_n\times n}.
    \]
    Let $\Hm=(h_{j,t})_{j=1,t=1}^{n-k,n} \in \F_{q^m}^{(n-k)\times n}$ a parity-check matrix of $\C$.\\  
    $$c \in RGSS_W(\C) \Longleftrightarrow \left\{\begin{array}{@{}l@{}}
        c=(\beta_1, \beta_2,..., \beta_{s_n})\Um \\
        (\beta_1, \beta_2,..., \beta_{s_n})\Um \Hm^T= 0
      \end{array}\right.\,.
    $$
    Let $(\gamma_1,...,\gamma_m)$ a $\F_q-$basis of $\F_{q^m}$.
    For all $i \in \{1, \cdots s_n\}$,  $j\in \{1, \cdots n-k\}$ and  $t \in \{1, \cdots n\}$, we have 
    \[
        \beta_ih_{j,t}= \sum\limits_{k=1}^m \delta_{i,t}^{(j,k)}\gamma_k \text{ where } \delta_{i,t}^{(j,k)} \in \F_q.
    \] 
    So,
       \[
       \forall j=1,...,n-k, \forall  k=1,...,m, \sum\limits_{i=1,t=1}^{s_n,n} \delta_{i,t}^{(j,k)}u_{i,t}= 0
       \]
    It is a linear system in $\sum\limits_{i=1}^{n}s_i$ unknowns and $m(n-k)$ equations. Therefore, the space of solution has dimension at least $\sum\limits_{i=1}^{n}s_i-m(n-k)$. Therefore, $$q^{\sum\limits_{i=1}^{n}s_i - m(n-k)} \leq |RGSS_W(\C)|.$$ \qed
\end{proof}

Notice that taking $s_1=...=s_n$ leads to the result obtained by E. Gabidulin and P. Loidreau in \cite{GabLoidreau2008} as a special case.
\begin{corollary}[\cite{GabLoidreau2008}, proposition 1]\\
Let $\mathcal{C}$ be a $[n, k, d=d_R(\C)]_{q^m}-$code and $V$ be a $\F_q-$subspace of $\F_{q^m}$ with dimension $s \leq m$.
If $ns - m(n-k) > 0$ then:
$$q^{ns - m(n-k)} \leq |\C \cap V^n|$$
\end{corollary}

\section{Generalized Subspace Subcodes of Gabidulin codes} \label{sec:gab}
\subsection{Gabidulin codes}
Gabidulin codes \cite{gabidulin1985theory} are defined by linear polynomials that were studied for the first time by O. Ore in \cite{Ore1933}.
These polynomials are of the form $P(z)=\sum_i p_iz^{q^i}$, where $p_i \in \F_{q^m}$ and $p_i \neq 0$ for finitely many i. By convention, we write $[i]:= q^i$, so $P(z)$ becomes $P(z)=\sum_i p_iz^{[i]}$. If $P \neq 0$, its $q-$degree is $deg_q(P)=max(\{i : p_i \neq 0\})$. Considering addition and multiplication by a scalar, we notice that  $q-$polynomials are $\F_q-$linear maps from $\F_{q^m}$ to $\F_{q^m}$.
Endowing $q-$polynomials with the addition and composition of polynomials allows getting a non-commutative ring structure that we denote by $\mathcal{P}_{q^m}$.
\begin{definition}
Let $\gv =(g_{1},g_{2},...,g_{n}) \in \F_{q^m}^n$ so that the elements $g_1, \cdots, g_n$ are $\F_q-$linearly independent, and $k < n$. The $[n,k]-$Gabidulin code with support $\gv$, dimension $k$ and length $n$ is defined by
$$Gab_k(\gv)=\{ (P(g_1),..., P(g_n) ) \in\F_{q^m} : P \in \mathcal{P}_{q^m},~ deg_q(P)<k \}$$
\end{definition}

A generator matrix of the Gabidulin code $Gab_k(\gv)$ is given by
$$G=\begin{pmatrix} g_{1}^{[0]}& \cdots & g_{n}^{[0]}\\ 
                g_{1}^{[1]}& \cdots & g_{n}^{[1]}\\
		 \vdots & \cdots & \vdots\\ 
		g_{1}^{[k-1]}& \cdots &g_{n}^{[k-1]}\\ 
	\end{pmatrix}$$
Gabidulin codes also satisfy the Singleton bound, that is to say, the minimum rank distance of $Gab_k(\gv)$ is $d_R = n - k + 1$. Thus, it can correct up to $\tau_{max}=\lfloor\frac{n-k}{2}\rfloor$ errors.
One can remark that the above definition and notations
are very close to those of Reed-Solomon codes with the set of distinct elements replaced by a set of linearly independent elements, and the classical power $g_{j}^{i}$ replaced by the ``Frobenius power'' $g_{j}^{[i]}$.
It is well known that a parity-check matrix of the $[n,k,d]-$Gabidulin code $Gab_k(\gv)$ is of the form:
 $$\Hm=\begin{pmatrix} 
     h_{1}^{[0]}&\cdots& h_{n}^{[0]}\\ 
     h_{1}^{[1]}&\cdots& h_{n}^{[1]}\\
	 \vdots & \cdots & \vdots\\ 
	h_{1}^{[d-2]}&\cdots&h_{n}^{[d-2]}\\ 
	\end{pmatrix}$$
where $\hv = \left(h_1, \cdots h_n\right)$ is a rank $n$ codeword of $Gab_k(\gv)^{ \perp }.$

\subsection{Rank Generalized Subspace Subcodes of Gabidulin codes}
Let us now apply the results of the previous section to the specific case of Gabidulin code. We start by the following proposition, providing upper and lower bounds for the cardinal of Generalized Subspace Subcodes of Gabidulin codes.

\begin{lemma}\label{unique}
    Let $Gab_k(\gv)$ be a $[n, k, d]-$Gabidulin code of length $n$ over $\mathbb{F}_{q^m}$. Let $V_1,...,V_n$ be a set of $n$ $\F_q-$subspaces of $\F_{q^m}$ with dimensions respectively $s_1, \cdots, s_n \leq m$. Set  
$W={\displaystyle \prod_{i=1}^{n} V_i}$, constituted of n-tuples with the $i-$th coordinate in $V_i$. 
Then, each codeword of $RGSS_W(Gab_k(\gv))$ maps to a unique codeword of $\mathcal{B}_W$, the Gabidulin code having parity-check matrix $T=(\beta_{j}^{[m-i+1]})_{i=1,j=1}^{d-1,s}$ with $s=\max\limits_{1\leq i \leq n}(s_i)$ and $(\beta_1, \beta_2,..., \beta_{s}) \in \F_{q^m}^s$ linearly
independent.
\end{lemma}
\begin{proof}
According to Proposition \ref{prop1}, we have 
$$c \in RGSS_W(\C) \Longleftrightarrow \left\{\begin{array}{@{}l@{}}
    c=(\beta_1, \beta_2,\cdots, \beta_{s_n})\Um \\
    (\beta_1, \beta_2,\cdots, \beta_{s_n})\Um \Hm^T= 0
  \end{array}\right.$$
where $\Hm=\begin{pmatrix} 
     h_{1}^{[0]}&\cdots& h_{n}^{[0]}\\ 
     h_{1}^{[1]}&\cdots& h_{n}^{[1]}\\
	\vdots & \cdots & \vdots\\ 
	h_{1}^{[d-2]}&\cdots&h_{n}^{[d-2]}\\ 
	\end{pmatrix}$.\\\\
	By setting $\Lm=\Um \Hm^T$, we have 
	$$\Lm_{i,j}=\sum\limits_{l=1}^{n}u_{i,l}h_l^{[j-1]}=\left( \sum\limits_{l=1}^{n}u_{i,l}h_l \right)^{[j-1]}=\Lm_{i,1}^{[j-1]}.$$
	Since $0=b_n\Um\Hm^T$, for $w_i=\Lm_{i,1}$ and $1\leq j \leq d-1$ we have \\ 
 \[
 \begin{tabular}{ccl}
     $0$ & $=$ & $(\beta_1, \beta_2,\cdots, \beta_{s_n}).\begin{pmatrix} 
     w_{1}^{[j-1]}\\ 
     w_{2}^{[j-1]}\\
	\vdots\\ 
	w_{s_n}^{[j-1]}\\ 
	\end{pmatrix}$ \\
      & $=$ & $\beta_1w_{1}^{[j-1]}+ \beta_2w_{2}^{[j-1]}+\cdots+\beta_{s_n}w_{s_n}^{[j-1]}$\\
      &  $=$ & $(\beta_1^{[m-j+1]}w_{1}+ \beta_2^{[m-j+1]}w_{2}+\cdots+\beta_{s_n}^{[m-j+1]}w_{s_n})^{[j-1]}$ \\
      &   $=$  & $\beta_1^{[m-j+1]}w_{1}+ \beta_2^{[m-j+1]}w_{2}+\cdots+\beta_{s_n}^{[m-j+1]}w_{s_n}$ \\
      &  $=$   & $(w_1,\cdots,w_{s_n}) \begin{pmatrix} 
     \beta_{1}^{[m-j+1]}\\ 
     \beta_{2}^{[m-j+1]}\\
	\vdots\\ 
	\beta_{s_n}^{[m-j+1]}\\ 
	\end{pmatrix}$
 \end{tabular} 
 \]
Which leads to  
$$(w_1,...,w_{s_n}) \begin{pmatrix} 
     \beta_{1}^{[m]}  & \beta_{1}^{[m-1]} & \cdots & \beta_{1}^{[m-d+2]}\\ 
     \beta_{2}^{[m]}  & \beta_{2}^{[m-1]} & \cdots & \beta_{2}^{[m-d+2]}\\
	\vdots            & \vdots            & \cdots & \vdots\\
	\beta_{s_n}^{[m]} & \beta_{s_n}^{[m-1]} & \cdots & \beta_{s_n}^{[m-d+2]}\\ 
	\end{pmatrix}=0$$
Let $\mathcal{B}_W$ be the Gabidulin code having parity-check matrix ~$\Tm=(\beta_{j}^{[m-i+1]})_{i=1,j=1}^{d-1,s_n}.$ Then, 
$c=b_n\Um \in RGSS_W(Gab_k(\gv)) \Longrightarrow (w_1,\cdots,w_{s_n}) \in \mathcal{B}_W$.
\end{proof}

\begin{theorem}\label{bounds}
Let $Gab_k(\gv)$ be a $[n, k, d]-$Gabidulin code of length $n$ over $\mathbb{F}_{q^m}$. Let $V_1,...,V_n$ be a set of $n$ $\F_q-$subspaces of $\F_{q^m}$ with dimensions respectively $s_1, \cdots s_n \leq m$. Set  
$W={\displaystyle \prod_{i=1}^{n} V_i}$, constituted of n-tuples with the $i-$th coordinate in $V_i$.
If $\sum\limits_{i=1}^{n}s_i - m(n-k) > 0$ then:
$$q^{\sum\limits_{i=1}^{n}s_i - m(n-k)} \leq |RGSS_W(Gab_k(\gv))| \leq q^{m(\max\limits_{1\leq i \leq n}(s_i)-d+1)}$$
\end{theorem}
\begin{proof}~\\
According to lemma \ref{unique}, 
$c=b_n\Um \in RGSS_W(Gab_k(\gv)) \Longrightarrow (w_1,\cdots,w_{s_n}) \in \mathcal{B}_W \Longrightarrow |RGSS_W(Gab_k(\gv))|\leq  |\mathcal{B}_W| = q^{m(s_n-d+1)}=q^{m(s_n-d+1)}=q^{m(\max\limits_i(s_i)-d+1)}$.\\
According to proposition \ref{prop1}
$$q^{\sum\limits_{i=1}^{n}s_i - m(n-k)} \leq |RGSS_W(Gab_k(\gv))|\leq q^{m(\max\limits_{1\leq i \leq n}(s_i)-d+1)}$$ \qed
\end{proof}
One can remark that taking $s_1 = \cdots = s_n = s$ gives the following corollary, which is also a known result from \cite{GabLoidreau2008}.
\begin{corollary}[\cite{GabLoidreau2008}]
Let $Gab_k(\gv)$ be a $[n, k, d]-$Gabidulin code of length $n$ over $\mathbb{F}_{q^m}$ and $V$ be a $\F_q-$subspaces of $\F_{q^m}$ with dimension $s \leq m$.
If $ns - m(n-k) > 0$ then,
$$q^{ns - m(n-k)} \leq |Gab_k(\gv)_{|V}| \leq q^{m(s-d+1)}$$
\end{corollary}
The following result is a direct consequence of Proposition \ref{ExtGenMat}, and will be useful to compute a generator matrix of a Generalized Subspace Subcode of a Gabidulin code.
\begin{proposition}
The $q-$ary image $Im_{q,(B_i)_i}(Gab_k(\gv))$ of a Gabidulin code $Gab_k(\gv)$ from the basis $B$ with respect to the $n$ bases $(B_1,..., B_n)$ of $\F_{q^m}$,  is generated by the block matrix $((\Mm_{g_j}^{q^{i-1}}))_{1 \leq i\leq k, 1 \leq j\leq n} \in (\F_{q}^{m\times m})^{k\times n}$ where the matrix $\Mm_{g_{j}}^T=\mathcal{M}_{B,B_j}(\theta_{g_{j}})$ is the matrix of the linear map $\theta_{g_{j}} = Exp_{B_j} \circ f_{g_{j}} \circ Exp^{-1}_{B}$ with $f_{g_{j}}(x)=g_{j}x$ for any $x \in \F_{q^m}$. 
\end{proposition}
We then have the following algorithm \ref{alg:cap} to compute a generator matrix of the shortened of a $q-$ary image of a generalized subcode on the subspaces of a Gabidulin code. 
\begin{algorithm}[H]
\caption{Generator matrix of $RGSS_W(Gab_k(\gv))$}\label{alg:cap}
\textbf{Input} A generator matrix $G \in \F_{q^m}^{k\times n}$ of $Gab_k(\gv)$, a set of $n$ bases $D_1,...,D_n$ of subspaces from $\F_{q^m}$ with dimensions respectively $s_1, \cdots, s_n \leq m$. \\

\textbf{Output} $G_U$ generator matrix of $Im_{q,(B_j)_j}^B(RGSS_W(Gab_k(\gv)))$
\begin{algorithmic}
\State $U=\bigcup\limits_{j=1}^n u_j$ such that for $1 \leq j\leq n$, $u_j=\{(j-1)m+1,..., (j-1)m+s_j\}$.
\State Complete the families $D_1,...,D_n$ into $n$ basis $B_1,...,B_n$ of $\F_{q^m}$
\State Set $\Qm_j$ be the change-of-basis matrix from the basis $B$ to the basis $B_j$, 
\State Compute a generator matrix $\hat{\Gm}_B$ of $Im_{q,B}(Gab_k(\gv))$
\State Compute a parity check matrix $\hat{\Hm}_B$ from $\hat{\Gm}_B$
\State Compute $\hat{\Hm}^B_{(B_j)_j}=\hat{\Hm}_B\cdot Diag((\Qm_1^{-1})^T,...,(\Qm_n^{-1})^T)$
\State Compute $\Mm_U=Punct_{I_U}(\hat{\Hm}^B_{(B_j)_j})$ as a parity-check matrix of 
 $S_U ( Gab_k(\gv))$ 
\State Compute $G_U$, generator matrix of $Im_{q,(D_j)_j}^B(RGSS_W(Gab_k(\gv)))$ from $\Mm_U$.
\end{algorithmic}
\end{algorithm}

\section{Parent Codes of Gabidulin RGSS codes}\label{sec:parent}
For $s_n=\max\limits_{1\leq i \leq n}(s_i) \geq d$, according to section 4, each codeword $$c=(\beta_1, \beta_2,\cdots, \beta_{s_n})\Um \in RGSS_W(Gab_k(\gv))$$ maps to a unique codeword $(w_1,\cdots,w_{s_n}) \in \mathcal{B}_W$ the Gabidulin code having parity-check matrix $T=(\beta_{j}^{[m-i+1]})_{i=1,j=1}^{d-1,s_n}$ which is MRD code with parameters $[s_n, s_n - d + 1, d]$.\\
For decoding needs, we may use the code $\mathcal{B}_W$, linked to $RGSS_W(Gab_k(\gv))$. 
\begin{definition}(Parent Code)\\
The $\F_{q^m}-$linear code $\mathcal{B}_W$ defined above, with parity-check matrix $$T=(\beta_{j}^{[m-i+1]})_{i=1,j=1}^{d-1,s_n}$$ is called  the parent code of $ RGSS_W(Gab_k(\gv))$, and is denoted by\\ $P_{RGSS_W(Gab_k(\gv))}$.
\end{definition}
\begin{proposition}
Let  $Gab_k(\gv)$ be a $[n, k, d]_{q^m}-$Gabidulin code having a parity check matrix admitting $\hv =(h_{1},h_{2},...,h_{n})$ as its first row. Let $V_1,...,V_n$ be a set of $n$ $\F_q-$subspaces of $\F_{q^m}$ with dimensions respectively $s_1, \cdots, s_n \leq m$ and $W = \prod_{i = 1}^{n} V_i $. For  $1 \leq i \leq n$, let $\bv_i = \{ \beta_1, \beta_2,..., \beta_{s_i} \}$ be a basis of $V_{i}$ such that the bases form an inclusion chain and $\bv$ the maximal basis up to isomorphism. Then the map
\[
\begin{tabular}{lccc}
  $f_\bv$ :    &   $W$  &$\longrightarrow $ & $(\F_{q^m})^{\max\limits_i(s_i)}$ \\
           &  $c = \bv \Um$   &  $\longmapsto$  & $f_\bv(c)= h\Um^T$
\end{tabular}
\]
is a rank-preserving injective $\F_q-$linear map and we have  
\[
f_\bv(RGSS_W(Gab_k(\gv)) \subset P_{RGSS_W(Gab_k(\gv))}
\]
Where $\Um=(u_{i,t})_{i=1,t=1}^{\max\limits_i(s_i),n} \in \F_{q}^{\max\limits_i(s_i)\times n}$ such that $$c_i=u_{1i}\beta_1+...+u_{s_ii}\beta_{s_i}+0\beta_{s_i+1}+...+0\beta_{\max\limits_i(s_i)}$$
\end{proposition}

\begin{proof}
    Since the components of $\hv$ are linearly independent over $\F_q$, then $ker(f_\bv)$ is $\{ 0 \}$. So, $f_\bv$ is injective and $Rk(b\Um)=Rk(\Um)=Rk(h\Um)$ for any $\Um$ in $\F_{q}^{\max\limits_i(s_i)\times n}$. 
    By construction, we have $f_\bv(RGSS_W(Gab_k(\gv))\subset P_{RGSS_W(Gab_k(\gv))}$. We can therefore note that the minimum distance of $RGSS_W(Gab_k(\gv)$ is smaller than that of $P_{RGSS_W(Gab_k(\gv))}$. \qed
\end{proof}
In \cite{Gabssrang2005}, the authors used this technique to make an algorithm for the encoding and decoding of the subcode on the subspace because there was no simple method to do it. Unfortunately, this method offers the same correction capacity as that of primary code. So we prefer to use the construction described in the previous section, which allows us to have a direct construction method even if we still do not have an improvement for the correction capacity.

\section{Applications to Cryptography}\label{sec:application}
Cryptography based on rank-metric codes is a very serious alternative to reduce key sizes in code-based cryptography, because the best attacks in rank-metric are exponential with a quadratic exponent, while the best in hamming metric are exponential with a linear exponent. 
The goal of this section is to show that these codes are potential candidates for rank-metric cryptography in a McEliece settings.
Let us recall that the general idea in a McEliece-like cryptosystems is to first choose a linear code $\C$ from a family of structured codes, that will serve as the secret key. The code $\C$ will then undergo some transformations in order to hide its structure and will result in a public code $\C_{pub}$ that will be published together with a correction capacity $t'$ depending on the used transformations and the correction capacity of the secret code $\C$. It is well known that, faced to such a system, the attacker must either distinguish the public code from a random code, or solve and instance of the general decoding problem with a random code.
Following this idea, Gabidulin, Paramonov and Tretjakov (GPT) \cite{gabidulin1991ideals} were the first to propose the use of the family of Gabidulin codes in a system today known as GPT cryptosystem. 
Unfortunately, a polynomial-time attack was proposed on the GPT cryptosystem
and its improvements by Overbeck \cite{overbeck2005new} and several other works. The weaknesses mainly come from the fact that Gabidulin codes are invariant under the Frobenius automorphism.
To avoid this type of attacks, some directions have been taken, and one of them is the use of subcodes of Gabidulin codes. The latter idea was followed by Berger, Gaborit and Ruatta \cite{10.1007/978-3-319-71667-1_13} with a system based on subcodes of $q-$ary images $Im_{q,B}(Gab_k(\gv))$ of Gabidulin codes for a fixed basis $B$. The security of the system is then relied on the subcode equivalence problem \cite{10.1007/978-3-319-71667-1_13}. 

According to \cite{10.1007/978-3-319-71667-1_13} the complexity  for solving the subcode equivalence problem by enumeration of Basis $B$ is a factor of the number of $\F_q-$bases in $\F_{q^m}$.
Therefore, Generalized Subspace Subcodes can be used to improve the security or the key-sizes of the cryptosystem in \cite{10.1007/978-3-319-71667-1_13} as using GSS codes of Gabidulin codes is equivalent to base the system on subcodes of $q-$ary images $Im_{q,(B_j)_j}(Gab_k(\gv))$ of Gabidulin codes for a family of bases $B_1, \cdots, B_n$ and, in our case, the attacker must search for $n$ bases instead of just one for solving the subcode equivalence problem by enumeration of Basis. More formally, let's consider $\Qm, \Qm_1, \cdots, \Qm_n \in GL(q,m)$. For a $q-$ary matrix code $\D$ \footnote{Here the code $\D$ can be viewed as included in $\F_q^{mn}$, but we use the term matrix code here to make connexion with \cite{10.1007/978-3-319-71667-1_13} where the codewords of $\D$ are viewed as matrices from $\F_q^{m \times n}$, and it is also the case here} we also define $\Phi_{(\Qm_j)_{j=1}^n,\Qm}(\mathcal{D})$ by
\begin{equation}\label{gsic}
\Phi_{(\Qm_j)_{j=1}^n,\Qm}(\mathcal{D})=\{ (\Qm_1 M_1, \cdots, \Qm_n M_n)Q ~|~ M=M_{\phi_B^{-1}}(\cv),~ \cv \in \mathcal{D}\}.
\end{equation}
Using Generalized Subspace Subcodes of Gabidulin codes in a McEliece-like cryptosystem results to the following problem.
\begin{problem}(Generalized Subcode equivalence on binary Image Codes, GSIC)\label{prob1}.\\
Given two $q-$ary matrix codes $\C$ and $\mathcal{D}$, find $\Qm, \Qm_1, \cdots, \Qm_n \in GL(q,m)$, such that $\Phi_{(\Qm_j)_{j=1}^n,\Qm}(\mathcal{D})$ is a subcode of $\C$.
\end{problem}
One can remark that for $\Qm_1=\Qm_2=\cdots = \Qm_n$, the previous problem gives rise to the SEMC problem defined in \cite[Theorem 3]{10.1007/978-3-319-71667-1_13}, and shown to be NP-complete.

One of the best ways to solve the SEMC as shown in \cite{10.1007/978-3-319-71667-1_13}, is to enumerate the set of all $\F_q$-bases in $\F_{q^m}$ to find $\Qm_1$ whereas in our case we need to find $n$ different bases $\Qm_1, \Qm_2,..., \Qm_n$ and, these will give a cost estimated as 
\[
\frac{m^nN_{q,m,n}}{m(q^m-1)}\times (n+\frac{q-1}{q})m^3\times k'm^2(n-k)
\]
instead of  
\[ 
\frac{N_{q,m,1}}{m(q^m-1)}\times (n+\frac{q-1}{q})m^3\times k'm^2(n-k)
\] 
as obtained in \cite{10.1007/978-3-319-71667-1_13}, where 

\[
N_{q,m,n}=(\prod\limits_{i=0}^{m-1}(q^m-q^i))^n=N_{q,m,1}\cdot(\prod\limits_{i=0}^{m-1}(q^m-q^i))^{n-1}.
\]
We emphasize that the idea of this section was not to propose a cryptosystem, because this requires a complete overview of the different existing attacks and this goes beyond the scope of this paper. However, we wanted to at least show that generalized subspace subcodes in the rank-metric deserve to be studied for their potential applications in cryptography. We see from the above that there could be a significant gain by using RGSS of Gabidulin codes in a McEliece settings compared to the system of Berger et al. \cite{10.1007/978-3-319-71667-1_13}. 


\section{Conclusion}\label{sec:conclusion}
In conclusion, we have generalized the notion of Subspace Subcodes in Rank metric introduced by Gabidulin and Loidreau and also characterize this family by giving an algorithm allowing to compute  generator and parity-check matrices of these codes from the generator and parity-check matrices of the associated extended codes. We have also studied the specific case of generalized subspace subcodes of Gabidulin codes and show that they are applicable to cryptography in a McEliece settings. 

The proposed algorithm for generating generator and parity-check matrices based on associated extended codes not only enhances the practicality of these codes but also opens avenues for further exploration and application in cryptology.



\bibliographystyle{plain}
\bibliography{code}
	
\end{document}